\documentclass[reqno, draft]{amsart}
\usepackage{amssymb, amscd,	amsfonts}
\usepackage[mathscr]{eucal}
\usepackage{graphicx}
\usepackage{amscd}

\newtheorem{theorem}{Theorem}[section]

\newtheorem{corollary}[theorem]{Corollary}
\newtheorem{proposition}[theorem]{Proposition}

\theoremstyle{definition}
\newtheorem{definition}[theorem]{Definition}

\newcommand{\restrict}{\,{\mathbin{\vert\mkern-0.3mu\grave{}}}\,}
\newcommand{\remove}[1]{}
\newcommand{\luk}{\L u\-ka\-s\-ie\-wicz}

\DeclareMathOperator{\I}{[0,1]}
\DeclareMathOperator{\cube}{[0,1]^{\it n}}
\DeclareMathOperator{\cantor}{\{0,1\}^{\it n}}

\DeclareMathOperator{\Form}{\mathsf{FORM}}
\DeclareMathOperator{\Val}{\mathsf{VAL}}
\DeclareMathOperator{\var}{\rm var}

\newcommand{\toe}{^{\langle e \rangle}}
\newcommand{\firste}{\frac{1}{e+1}}
\newcommand{\laste}{\frac{e}{e+1}}

\title[Drawing sound conclusions from unsound premises]
{Drawing sound conclusions from  unsound
premises}

\author[D. Mundici]{Daniele Mundici$^\dag$}
\address[D. Mundici]{Department of Mathematics ``Ulisse Dini'' \\
University of Florence \\
viale Morgagni 67/A \\
50134 Florence \\
Italy}

\email{mundici@math.unifi.it }

\author[C. Picardi]{Claudia Picardi $^\ddag$}
\address[C. Picardi]{Department of Computer Science,  University of Turin \\
  Corso Svizzera 185,  Turin \\
Italy }
\email{claudia.picardi@di.unito.it }



  \date{\today}

\begin{document}

\begin{abstract}
Given sets $\Phi_{1}=\{\phi_{11},\dots,\phi_{1u(1)}\},
\ldots,\Phi_{z}=\{\phi_{z1},\dots,\phi_{zu(z)}\}$
of boolean  formulas,  a
formula $\omega$ follows from the conjunction 
$\bigwedge\Phi_i=
\bigwedge  \phi_{ij}$  iff
$\neg \omega\wedge \bigwedge_{i=1}^z \Phi_i$ is unsatisfiable.
Now assume that,  given integers $0\leq e_i < u(i)$,
 we must check if
$\neg \omega\wedge \bigwedge_{i=1}^z \Phi'_i$
remains  unsatisfiable, where    $\Phi'_i\subseteq
\Phi_i$ is obtained by deleting
$\,\,e_{i}$ arbitrarily chosen
formulas of $\Phi_i$, for  each
$i=1,\ldots,z.$
Intuitively,  does $\omega$
{\it stably}  follow,
after removing $e_i$  random formulas from
each  $\Phi_i$?
We construct a quadratic
  reduction   of this problem
to the consequence problem in infinite-valued
\luk\ logic \L$_\infty$.
% in such a way that
% for
% every instance  $I=(\Phi_{1},\ldots,\Phi_{z}, \omega; e_1,\ldots,e_z)$
%of the problem,  the length  of
%$\rho(I)$ is proportional to the length of   $I$.
In this way we obtain
a self-contained proof that  the \L$_\infty$-consequence problem
  is coNP-complete.
    \end{abstract}

\maketitle

%%%%%%%%%%%%%%%%%%%%%%%%%%%%%%%%%%%%%%%%%%%%%%
\section{Foreword}
%%%%%%%%%%%%%%%%%%%%%%%%%%%%%%%%%%%%%%%%%%%%%%

Throughout, boolean formulas
are  strings on the alphabet
$  \{ X\,, |\, , \,\neg\,,\,
  \wedge\,,\, \vee, \,  )\, , \, (  \,\, \},$
  as given by the usual syntax of propositional
  logic.
  Strings of the form
$X|, X||,  \ldots$    {\rm (for\;
short} $X_1, X_2, \ldots),$
are called {\it variables}.

\medskip
The {\it  Stable Consequence problem}
% {\rm \stu} $\subseteq {\mathcal A}^{*}$,
is as follows:

\vspace{0.1cm}
\noindent INSTANCE: A finite list
$\Phi_1,\ldots, \Phi_k$  of finite sets $\Phi_i$ of
boolean formulas,  and for each
$i=1,\ldots,k$  an integer
$0 \leq e_i <
{\rm card}(\Phi_i)
=
\mbox{ number of elements of }
\Phi_i.$

\smallskip
\noindent QUESTION:
Letting for each  $i=1,\ldots,k$ the
set $\Phi'_i\subseteq \Phi_i$ be
obtained by arbitrarily deleting $e_i$ formulas
of   $\Phi_i$,  is  the
  conjunction of all
formulas in $\Phi'_1 \cup \ldots \cup \Phi'_k$
  unsatisfiable ?

\smallskip
The  problem introduced in the abstract is
the special case of  the Stable Consequence problem
for $\Phi_{1}=\{\neg\omega\}$ and $e_1=0.$
The complementary 
  problem also generalizes the
decision version of the Maximum Satisfiability
problem.
Its significance
will be discussed in Section \ref{section:remarks}.

\smallskip
  In Section \ref{section:final}
  (Theorem \ref{theorem:final} and Corollary \ref{corollary:final})
   we construct a polynomial time reduction  $\rho$
of  the  Stable Consequence problem
  to the consequence problem in
  \luk\ infinite-valued logic \L$_\infty$.
Specifically,  every instance
  $$I=(\Phi_1,\ldots, \Phi_k; \, e_1,\ldots,e_k)$$
  of   the  Stable Consequence problem
  is transformed by $\rho$
   into a pair  $\rho(I)=
  (\theta,\phi)$ of  \L$_\infty$-formulas
  in such a way that $I$ belongs to the
  Stable Consequence problem
   iff $\phi$ is a consequence of $\theta$
  in \L$_\infty$. All preliminary material
  on \L$_\infty$-consequence is collected in
  Section 2.  Building on Sections 3 and 4, 
  Proposition \ref{proposition:odot} 
explains how the numerical parameters  $e_i$
are   incorporated
into formulas of \L$_\infty.$

Letting $v_I$ be the number of distinct variables in
$I$, and $|I|$ be the {\it  length} (i.e., the number
of occurrences of symbols)  of $I$, it turns out that
$$
|\rho(I)| < c\cdot v_I\cdot |I|,
$$
for some constant $c$ independent of $I$.
   Further,   $I$  and  $\rho(I)$  have the same variables.
A self-contained proof of
the  coNP-completeness of the consequence
  problem in  \L$_\infty$ in finally obtained in
 Corollary \ref{corollary:remake}.

This
strengthens \cite[Theorem 9.3.4]{cigdotmun},
as well as    \cite[Theorem 18.3]{newmvbook},
and solves Problem 5.3 in \cite{jer}.

  We refer to \cite[\S 4]{cigdotmun}  for background
  on \luk\ propositional logic \L$_\infty,$ and to
  \cite[\S 7]{macyou}  for polynomial time reducibility
  and NP-completeness.

%%%%%%%%%%%%%%%%%%%%%%%%%%%%%%%%%%%%%%%%%%%%%%%%%%%%%%
\section{Consequence in infinite-valued \luk\ logic}
%%%%%%%%%%%%%%%%%%%%%%%%%%%%%%%%%%%%%%%%%%%%%%%%%%%%%

To efficiently write down     \L$_\infty$-formulas
it will be convenient to
  use the richer alphabet
$  \{\, X, |, \neg,  \odot,
\oplus, \wedge, \vee,  ) ,  (  \; \}.$
The symbols $\neg, \odot, \oplus$
are called the {\it negation, conjunction,
{\rm and} disjunction} connective,
respectively.  We call $ \wedge$ and $ \vee$
the {\it idempotent} conjunction
and disjunction.
As shown in \cite[(1.2),  1.1.5]{cigdotmun}, the connective
$\odot$, as well as
the idempotent connectives are definable
in terms of $\neg$ and $\oplus,$ in the  sense of
  (\ref{equation:valuation1})-(\ref{equation:valuation3})
    below.
Following \cite[(4.1)]{cigdotmun},  we write
$\alpha\to\beta$  as an abbreviation of  $\beta\oplus \neg \alpha$.
Further,   $\alpha\leftrightarrow \beta$  stands for
$(\alpha\to \beta)\odot(\beta\to \alpha).$

To increase readability
we  assume that the negation connective  $\neg$ 
  is more binding
than $\odot$, and the latter is
more binding than $\oplus$; the idempotent
connectives $\vee$ and $\wedge$ are less
binding than any other connective.

For each
  $n=1,2,\ldots,$ we let ${{\mathsf{FORM}}}_{n}$
denote the set of formulas
$\psi(X_{1},\ldots,X_{n})$ whose variables are contained in the set
$\{X_{1},\ldots,X_{n}\}$.
More generally, for any   set $\mathcal X$ of variables,
$\mathsf{FORM}_{\mathcal X}$
denotes the set of formulas whose variables
are contained in  $\mathcal X.$
For each formula  $\phi$  we let  ${\rm var}(\phi)$ be the set of
variables occurring in $\phi$.

For any  formula $\phi\in {{\mathsf{FORM}}}_{n}$ and
integer $k = 1,2,\ldots,$ the iterated
conjunction $\phi^k$ is defined by
\begin{equation}
\label{equation:iterated}
\phi^{1} = \phi, \,\, \,\, \phi^{2} = \phi \odot \phi,
\,\,\,\,  \phi^{3} = \phi \odot \phi \odot \phi,\ldots.
\end{equation}
The iterated disjunction
$\;k\centerdot \phi\;$ is defined by
\begin{equation}
\label{equation:iterated-disj}
1\centerdot \phi = \phi, \,\,\,\,
  2\centerdot \phi = \phi \oplus \phi, \,\,\,\, 3\centerdot \phi
= \phi \oplus \phi \oplus \phi,\ldots .
\end{equation}

\begin{definition}
\label{definition:valuation}
A {\it valuation} \index{Valuation} (of ${\mathsf{FORM}}_{n}$ in
\L$_\infty$) is a function $V\colon {\mathsf{FORM}}_{n} \to [0,1]\,$
such that
$$
         {\it V}(\neg \phi) =  1- {\it V}(\phi), \,\,\,\,
{\it V} (\phi\oplus \psi) =  %V(\phi)\oplus V(\psi) \,\,\,=\,\,\,
\min(1,{\it V} (\phi)+{\it V}(\psi))
$$
and, for the derived connectives $\odot,\vee,\wedge,$
\begin{eqnarray}
\label{equation:valuation1}
{\it V} (\phi \odot \psi) &=&
  \max(0, {\it V} (\phi)+{\it V}(\psi)-1)
=
  V(\neg(\neg\phi\oplus \neg\psi))\\
  \label{equation:valuation2}
{\it V} (\phi\vee \psi) &=& %V(\phi)\vee V(\psi) \,\,\,=\,\,\,
\max({\it V} (\phi), {\it V}(\psi))
=
V(\neg(\neg \phi \oplus \psi)\oplus \psi)\\
\label{equation:valuation3}
{\it V} (\phi\wedge \psi) &=& %V(\phi)\wedge V(\psi) \,\,\,=\,\,\,
\min({\it V} (\phi), {\it V}(\psi))
=
V(\neg(\neg\phi \vee \neg \psi)).
\end{eqnarray}
\end{definition}

\medskip
We denote by $\mathsf{VAL}_n$ the
set of valuations of $\mathsf{FORM}_n$.
More generally, for any   set $\mathcal X$ of variables,
$\mathsf{VAL}_{\mathcal X}$
denotes the set of valuations
  $V\colon {\mathsf{FORM}}_{\mathcal X} \to [0,1]$.

The non-ambiguity of the syntax of \L$_{\infty}$
is to the effect that each
  $V\in \mathsf{VAL}_n$ is uniquely determined
by its restriction to  $\{X_1,\dots,X_n\}.$
Thus  for
every   point $x
= (x_1,\ldots,x_n)\in
\I^n$  there is a uniquely determined valuation
$V_x\in \mathsf{VAL}_n$    such that
\begin{equation}
\label{equation:doppiovu}
V_x(X_i)=x_i  \,\mbox{ for all  }
\, i=1,\ldots,n.
\end{equation}
Conversely, upon
identifying the two sets
$\I^n$ and $\I^{\{X_1,\ldots,X_n\}}$, we
can write  $ x=V_x\restrict\{X_1,\ldots,X_n\}.$

\smallskip
 
For any set $\Phi\subseteq {{\mathsf{FORM}}}_{\mathcal X}$
and $V\in \mathsf{VAL}_{\mathcal X}$
we say that $\,V$ {\it satisfies}
$\,\Phi$ if  $V(\psi)=1$ for all $\psi\in \Phi.$
A formula  $\phi$ is a {\it tautology}
if it is satisfied by all  valuations
$V\in \mathsf{VAL}_{\var(\phi)}.$

\begin{proposition}
\label{proposition:woj}
For all $n=1,2,\ldots$
and   $\theta,\phi\in \mathsf{FORM}_n$ the following conditions are
equivalent:
\begin{itemize}
     \item[(i)]
Every valuation  $V\in \mathsf{VAL}_n$
  satisfying  $\theta$ also satisfies $\phi$.
  In other words, $ \phi$ is  a 
  {\rm semantic  \L$_\infty$-consequence} of $\theta$;

\smallskip
      \item[(ii)] For some  integer
      $k>0$ the formula $  \theta^k\to \phi$ is a
tautology. (Notation of (\ref{equation:iterated})).

\smallskip
   \item[(iii)] For some  integer $k>0$
	the formula
\begin{equation}
     \label{equation:iterated-implication}
	\underbrace{\theta\to(\theta\to(\theta\to
	\cdots\to(\theta\to (\theta\to\phi))\cdots))}_{
	\mbox{\tiny $k$  occurrences of $\theta$}}
\end{equation}
	is a tautology.

  \smallskip
\item[(iv)] For some integer $k>0$ there is a sequence of formulas
$\chi_0,\ldots,\chi_{k+1}$ such that $\chi_0=\theta$,
$\chi_{k+1}=\phi,$ and for each $i=1,\ldots,k+1$
either  $\chi_i$ is a tautology, or there are
   $p,q \in\{0,\ldots,i-1\}$ such that
$\chi_{q}$ is the formula $\chi_p\to \chi_i$.

  \smallskip
			\item[(v)] For some integer $k>0$ there is a 
sequence of formulas
$\chi_0,\ldots,\chi_{k+1}$ such that $\chi_0=\theta$,
$\chi_{k+1}=\phi,$ and for each $i=1,\ldots,k+1$
either  $\chi_i$ is a tautology in $\mathsf{FORM}_n$, or there are
   $p,q \in\{0,\ldots,i-1\}$ such that
$\chi_{q}$ is the formula $\chi_p\to \chi_i$.
In other words,  $\phi$ is a {\rm syntactic \L$_\infty$-consequence} of $\theta$.

\end{itemize}
\end{proposition}

\begin{proof}  (ii)$\Leftrightarrow$(iii)  is promptly verified,
because the two formulas  (\ref{equation:iterated-implication}) and
       $  \theta^k\to \phi$
       %%
        % and
        % $$
        %   \underbrace{\theta\to(\theta\to(\theta\to
        %   \ldots\to(\theta\to (\theta\to\phi))\ldots))}_{
        %   \mbox{\tiny $k$  occurrences of $\theta$}}
        %   $$
        %%
       are equivalent in \L$_\infty$.
(iv)$\Leftrightarrow$(i)   follows from
        \cite[4.5.2, 4.6.7]{cigdotmun}.
     (iv)$\Leftrightarrow$(iii) follows from
     \cite[4.6.4]{cigdotmun}.
     (v)$\Rightarrow$(iv) is trivial.
     Finally, to prove
  (iii)$\Rightarrow$(v), arguing by
   induction on $k$,  one verifies  that
     $\phi$  can be obtained as the final formula  $\chi_{k+1}$ of
     a sequence $\chi_0,\ldots,\chi_{k+1}$  as in (v),
which only requires
  the assumed tautology
  (\ref{equation:iterated-implication}).
\end{proof}

We write $\theta\vdash \phi$  if  $\theta$ and $\phi$
satisfy the equivalent conditions   above,
 and we say  that
$\phi$ is an {\it \L$_\infty$-consequence} of $\theta$  
 without fear of ambiguity.

\smallskip
An instance of the
\L$_\infty$-{\it consequence problem}
%\cons
  is a pair of formulas  $(\theta, \phi)$.
  The problem asks
if $\phi$ is an
  \L$_\infty$-consequence of
$\theta$.

%%%%%%%%%%%%%%%%%%%%%%%%%%%%%%%%%%%%%%%%%%%%%%%%%%%%%%%%%%%%%%
\section{The function $\hat{\phi}$ associated to an \L$_\infty$-formula
$\phi$}
%%%%%%%%%%%%%%%%%%%%%%%%%%%%%%%%%%%%%%%%%%%%%%%%%%%%%%%%%%%%%%

\begin{proposition}
\label{proposition:tutto}
To every formula  $\phi = \phi(X_1,\ldots,X_n)
\in \mathsf{FORM}_{n}$
let us  associate  a function, denoted
$\widehat{\phi} \colon [0,1]^n \to [0,1],$
via the following   inductive
procedure: for all
$x = (x_1,\ldots,x_n)
\in [0,1]^n,$

\begin{eqnarray*}
\widehat{X_i}(x)
& = & x_i
   \;\; (i=1,\ldots,n),    \\
\widehat{\neg \psi}(x)
& = & 1- \widehat{\psi(x)},   \\
\widehat{\psi\oplus \chi}(x)
& = &
%\widehat{\psi}(x) \oplus \widehat{\chi}(x)\;\;=\;\;
\min(1, \widehat{\psi}(x) + \widehat{\chi}(x)),\\
\widehat{\psi\odot \chi}(x)
& = &
%\widehat{\psi}(x)\odot \widehat{\chi}(x)\;\;=\;\;
\max(0, \widehat{\psi}(x)
+ \widehat{\chi}(x) -1),\\
\widehat{\psi \wedge \chi}(x)
& = &
  %\widehat{\psi}(x) \wedge \widehat{\chi}(x)\;\;=\;\;
  \min(\widehat{\psi}(x), \widehat{\chi}(x)),\\
\widehat{\psi \vee \chi}(x)
& = &
%\widehat{\psi}(x) \vee \widehat{\chi}(x)\;\;=\;\;
\max( \widehat{\psi}(x), \widehat{\chi}(x)).
\end{eqnarray*}
Then generalizing
(\ref{equation:doppiovu})  we have the identity
     \begin{equation}
	\label{equation:evaluation}
    \hat{\phi}(x)= {V_x}(\phi)
      \,\,\,\mbox{ for all } \, x \in [0,1]^{n}.
     \end{equation}
\end{proposition}

\begin{proof}
Immediate by Definition
\ref{definition:valuation}, arguing
by  induction on the number of connectives in
$\phi$. It should be noted  that
the definition of $\hat\phi$
relies on  the non-ambiguity of the syntax of
\L$_\infty.$
     \end{proof}

\begin{proposition}
\label{proposition:due-sei}
For each  $n=1,2,\ldots,$  $\,e=2,3,\ldots$  and
  valuation $V \colon \Form_n \to [0,1],$
the following conditions are equivalent:

\smallskip
\begin{itemize}
\item[(i)]
$
V\,\, \mbox{satisfies}\,\,
\bigwedge_{i=1}^n   (X_i^e \leftrightarrow \neg X_i) \vee
(X_i \leftrightarrow \neg \,\,e\centerdot X_i). $
(Notation of (\ref{equation:iterated})-(\ref{equation:iterated-disj})).
%\mbox{ iff }
%\alpha\mbox{satisfies}
%\theta_e(X_1) \odot \cdots \odot \theta_e(X_n)

\medskip
\item[(ii)]  For each  $\,i=1,\ldots, n$, 
$\,\, V(X_i)\in  \left\{\frac{1}{e+1},\frac{e}{e+1}\right\}.
$
\end{itemize}
\end{proposition}

\begin{proof}
Let
$\xi_e$ be the  \L$_\infty$-formula
$X^e \leftrightarrow \neg X$, and $\widehat{\xi_e}\colon \I\to \I$
its associated function.
Recalling
(\ref{equation:evaluation})
and the definition of the $\leftrightarrow$
connective,
  for every $y\in [0,1],\,\,\,$
  we can write
$\;\widehat{\xi_e}(y)=1\;$ iff
$\;\widehat{X^e}(y)
=
1-y$.  Further, by induction on $e$, 
$$
\widehat{X^e}(y)  =
\underbrace{y\odot \cdots \odot y}_{e\,\,\,
  \mbox{\footnotesize times}}=
  \max(0, ey-e+1) =
\begin{cases}
0 & \mbox{ if }  \;0 \leq y < \frac{e-1}{e}  \cr
  ey-e+1 &  \mbox{ if }  \;\frac{e-1}{e} \leq y \leq 1 . \\
  \end{cases}
$$
Thus, $\widehat{\xi_e}(y) = 1\;$
iff
$\;ey -e + 1 = 1-y\;$
iff    $y =\laste$.
In other words,
a valuation satisfies
$X^e \leftrightarrow \neg X$ iff it
evaluates $X$ to $\frac{e}{e+1}.$

Similarly, letting
$\chi_e$ be the formula
$ X \leftrightarrow \neg \,\,e\centerdot X$
we obtain $\widehat{\chi_e}(y) =
\widehat{\xi_e}(1-y)$,
whence
$\widehat{\chi_e}(y)=1\;$
iff $\;\widehat{\xi_e}(1-y) = 1\;$
iff
$\;1-y =\laste\;$ iff
  $\;y =\firste.$
  Thus a valuation satisfies
  $ X \leftrightarrow \neg \,\,e\centerdot X$
  iff it evaluates  $X$ to $\firste.$

  Summing up, by
  (\ref{equation:valuation2})-(\ref{equation:valuation3}),
  a valuation
  satisfies
  $\bigwedge_{i=1}^n   (X_i^e \leftrightarrow \neg X_i) \vee
(X_i \leftrightarrow \neg \,\,e\centerdot X_i)$
iff it evaluates each  $X_i$ either to $\firste$ or to $\laste.$
\end{proof}

%%%%%%%%%%%%%%%%%%%%%%%%%%%%%%%%%%%%%%%%%%%%%%%%%%%%%%%%%%%%%%
\section{The $\ddagger$-transform of a boolean formula}
%%%%%%%%%%%%%%%%%%%%%%%%%%%%%%%%%%%%%%%%%%%%%%%%%%%%%%%%%%%%%%

As the reader will recall,
every  boolean formula  $\psi$  in this paper is constructed
  from the variables only using the connectives
  $\neg, \vee, \wedge.$
A boolean formula    is said to be
in {\it negation normal form}
if  the negation symbol can only precede
  a variable. Any boolean formula  $\psi$ can be
  immediately  reduced
  into an equivalent formula $\psi^\dagger$ in
   negation normal form by
   using De Morgan's laws to push negation inside
   all conjunctions and disjunctions,
   and eliminating double negations.
   The same variables occur in $\psi$ and
   $\psi^\dagger$.
Further, the number of occurrences of variables in
  $\psi$ is the same as in 
  $\psi^\dagger$.

  \begin{definition}
  \label{definition:ddagger}
  Let $\psi=\psi(X_1,\dots,X_n)$ be a boolean formula.
We denote by
$\psi^{\ddagger}$   the formula
   in \luk\ logic \L$_\infty$ obtained from $\psi$
   by the following procedure:

   \begin{itemize}
   \item[---]   write the negation normal form $\psi^\dagger$, and for
   each  $ i=1,\dots,n,$

   \item[---] replace every occurrence of $\neg X_i$ in
    $\psi^\dagger$ by the formula
    $X_i \vee \neg (X_i\odot X_i),$

   \item[---] and simultaneously replace every occurrence
    of the non-negated variable   $X_i$ by the formula
    $\neg X_i \vee (X_i\oplus X_i),\quad i=1,\dots,n.$
   \end{itemize}

\noindent
In other words,   the
{\it $\ddagger$-transform}
   $\psi^\ddagger$  of $\psi$ is the
   \L$_\infty$-formula defined by:
   \begin{eqnarray*}
    (\neg X_i)^{\ddagger}&=& X_i \vee \neg (X_i\odot X_i), \\
      X_i^{\ddagger}&=&\neg X_i \vee (X_i\oplus X_i),
  \,\,\, \mbox{ if $X_i$ is not preceded by $\neg$}
    \end{eqnarray*}
  and by induction on the number of binary connectives in $\psi^\dagger,$
   \begin{eqnarray*}
    (\sigma \wedge \tau)^{\ddagger} &=& \sigma^{\ddagger}
    \wedge \tau^{\ddagger}\\
     (\sigma \vee \tau)^{\ddagger} &=& \sigma^{\ddagger}\vee \tau^{\ddagger}.
   \end{eqnarray*}
   \end{definition}
  \unitlength0.65cm
\begin{picture}(5,5)

\multiput(0,0)(5,0){4}{\line(1,0){4}}
\multiput(0,0)(5,0){4}{\line(0,1){4}}
\multiput(0,4)(5,0){4}{\line(1,0){4}}
\multiput(4,0)(5,0){4}{\line(0,1){4}}

\thicklines

%PRIMO  x \oplus x
\put(0,0){\line(1,2){2}}
\put(2,4){\line(1,0){2}}
%\put(0.5,3.2){\footnotesize{$(l+1)X\odot \neg X$}}
\put(-1.4,2.9){\footnotesize{$\widehat{X\oplus X}$}}

%SECONDO not x
\put(5,4){\line(1,-1){4}}
%\put(0.5,3.2){\footnotesize{$(l+1)X\odot \neg X$}}
\put(4.2,3){\footnotesize{$\widehat{\neg X}$}}

%%TERZO
\put(11.3,2.69){\line(1,2){0.65}}
\put(10,4){\line(1,-1){1.30}}
\put(11.95,4){\line(1,0){2.05}}
\put(9.4,2.9){\footnotesize{$\widehat{X^{\ddagger}}$}}

%QUARTO
\put(15,4){\line(1,0){2}}
\put(17,4){\line(1,-2){0.65}}
\put(17.65,2.7){\line(1,1){1.3}}
\put(14.1,2.9){\footnotesize{$\widehat{\neg X^{\ddagger}}$}}

\put(2.8,-1.0){\mbox{\small
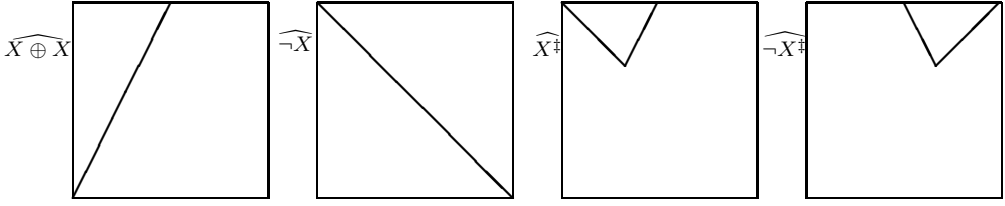
Figure 1.  The graphs of the functions $\widehat{X\oplus X}$,
$\widehat{\neg X}$, $\widehat{X^{\ddagger}}$}
and 	$\widehat{\neg X^{\ddagger}}$.}

\end{picture} 

\bigskip

\bigskip

\bigskip
\begin{definition}
\label{definition:transform}
     Fix $e=2,3,\ldots .$
For each
$y \in \{0,1\}$ we let  $y \toe$ be the only
point of $\I$ lying at a distance
$\frac{1}{e+1}$  from $y$.
%%
  % Thus,
  % $$
  % y \toe =
  % \begin{cases}
  %	\firste	&  \mbox{if}  \,\,\,\,y=0, \cr
  % \laste  &  \mbox{if}	 \,\,\,\, y=1.\cr
  % \end{cases}
  % $$
  %%
More generally,  for any
$x = (x_1, \ldots ,x_{m} ) \in
\{0,1\}^{m}$, the point  $x\toe\in \I^m$
is defined by
$
x\toe  = (x_1 \toe, \ldots ,x_{m}\toe ).
$
%and is called the {\it $e$-transform of $x$}.
\end{definition}

   \begin{proposition}
   \label{proposition:2223}
For  any boolean valuation  $W,$
   $$
   W\colon \{\mbox{boolean formulas in the
   variables}\,\,\, X_1,\dots,X_n\}
   \to \{0,1\},
   $$
    let $w\in \{0,1\}^{\{X_1,\dots,X_n\}}=\cantor$ be the restriction
  of W  to  the set $\{X_1,\dots,X_n\}$.  Then for every
boolean formula $\psi(X_1,\dots,X_n)$ and $e=2,3,\ldots$
   we have:
   \begin{eqnarray*}
   W\,\,\mbox{satisfies}\,\,\, \psi & \mbox{iff}&
   \;\widehat{\psi^{\ddagger}}(w\toe) = 1\\
   W\,\,\mbox{does not satisfy}\,\,\,\psi
   & \mbox{iff}&\;\widehat{\psi^{\ddagger}}(w\toe) = \laste.
   \end{eqnarray*}
\end{proposition}

\begin{proof}
Our assumption about $e$ ensures that
$0\toe < 1\toe.$
For each variable
$X$ we first prove  (see Fig. 1):

\smallskip
  \begin{itemize}
  \item[(i)] $\widehat{X^{\ddagger}}(\firste)
  =\laste$, \label{one:tauval:i}

  \smallskip
  \item[(ii)]  $\widehat{X^{\ddagger}}(\laste)
   =1$,\label{one:tauval:ii}

   \smallskip
  \item[(iii)]  $\widehat{\neg X^{\ddagger}}(\firste)
   =1$,\label{one:tauval:iii}

   \smallskip
  \item[(iv)]  $\widehat{\neg X^{\ddagger}}(\laste)
  =\laste$\label{one:tauval:iv}.
  \end{itemize}

\medskip
(i)--(ii)  By
(\ref{equation:evaluation}),
    for all $y \in [0,1]$
we can write
$\widehat{X^{\ddagger}}(y)
= \max(\widehat{\neg  X}(y),\widehat{X \oplus X}(y))
$
$
= \max(1-y,\min(1,2y)).
$
Thus, 
$$ \widehat{X^{\ddagger}}\left(\firste\right)   =
\max\left(\laste,\min(1,\frac{2}{e+1})\right) =
\max\left(\laste,\frac{2}{e+1}\right)=\frac{e}{e+1}
$$
and
$$
\widehat{X^{\ddagger}}\left(\laste\right)    =
\max\left(\firste,\min(1,\frac{2e}{e+1})\right)
= \max\left(\firste,1\right) = 1.
$$

\medskip
\medskip
(iii)--(iv)
Again by
(\ref{equation:evaluation}), we can write
$ \widehat{\neg X^{\ddagger}}(y)  =
\max(\widehat{X}(y),\widehat{\neg (X \odot X)}(y))
= \max(y, 1-\max(0, 2y-1)) = \max(y, \min(1,2-2y)),$
whence

\medskip
$$
\widehat{\neg X^{\ddagger}}\left(\firste\right)   =
\max\left(\firste, \min(1,2-\frac{2}{e+1})\right) =
\max\left(\firste,1\right)
= 1
$$
and
$$
\widehat{\neg X^{\ddagger}}\left(\laste\right)     =
\max\left(\laste,\min(1,2-\frac{2e}{e+1})\right)
=
\max\left(\laste,\frac{2}{e+1}\right) = \laste.
$$

\bigskip
Having thus settled (i)-(iv), the
  proof now proceeds by induction on the number $b$ of
binary connectives in $\psi^\dagger$, the equivalent
counterpart of $\psi$ in negation normal form.

\bigskip

\noindent{\it Basis, $b=0$.}
Then $\psi^\dagger\in \{X_i,\neg X_i\}.$

In case   $\psi^\dagger = X_i$ we have
\begin{eqnarray*}
&& W \mbox{ satisfies } \psi\\[0.08cm]
&\mbox{iff}&  W  \mbox{ satisfies } X_i, \mbox{ (because } 
\psi^\dagger  \mbox{ is
equivalent to } \psi )\\
&\mbox{iff}& w_i =1,  \mbox{  by definition of $w$}\\
&\mbox{iff}&   w_i\toe =\laste, \mbox{  by definition of $w_i{\toe}$}\\
&\mbox{iff}&  \widehat{X_i^{\ddagger}}(w_i\toe) =
\widehat{\psi^{\ddagger}}(w_i\toe)=1.
  \end{eqnarray*}
  The  $(\Downarrow)$-direction of the last bi-implication
  follows from (ii).
Conversely, for the $(\Uparrow)$-direction, if
  $w_i\toe\not=\frac{e}{e+1}$ then
  $w_i\toe=\frac{1}{e+1}$, whence by
  (i), $ \widehat{X_i^{\ddagger}}(w_i\toe) =
  \laste \not=1.$

%  $\psi^\dagger$ iff $w_i=1$ iff
%$w_i\toe =\laste$. Thus, by (i)-(ii),
%%Lemma \ref{lemma:tauval},
%$\widehat{\psi^\dagger}(w\toe)
%= \widehat{X_i^{\ddagger}}(\laste) = 1$.
%Conversely, $W$ does not satisfy $\psi^\dagger$ iff
%$w_i = 0$ iff  $w_i\toe=\firste$, whence
%again by (i)-(ii),
%%Lemma \ref{lemma:tauval},
% $\widehat{\psi^{\ddagger}}(w\toe)
% = \widehat{X_i^{\ddagger}}(\firste) =\laste$.
%
%
%

\smallskip

  The case  $\psi^\dagger = \neg X_i$ is similarly
  proved using (iii)-(iv).

\bigskip
\noindent{\it Induction step.}
Suppose  $\psi^\dagger = \sigma \wedge \tau$. Then
\begin{eqnarray*}
&& W \mbox{ satisfies } \psi\\
&\mbox{iff}&  W  \mbox{ satisfies } \psi^\dagger\\
&\mbox{iff}&  W  \mbox{ satisfies  both }   \sigma^\dagger
  \mbox{ and }
\tau^\dagger\\
&\mbox{iff}&  W  \mbox{ satisfies  both }   \sigma \mbox{ and }
\tau\\
&\mbox{iff}& \widehat{\sigma^{\ddagger}}(w\toe)
= \widehat{\tau^{\ddagger}}(w\toe) = 1,  \mbox{ by induction hypothesis.}
\end{eqnarray*}
Thus, if $W$ satisfies  $\psi$ then
$$\widehat{\psi^{\ddagger}}(w\toe)
=(\widehat{\sigma^{\ddagger}}
\wedge \widehat{\tau^{\ddagger}})(w\toe) =
\min(1,1)
=1.$$

Conversely,
\begin{eqnarray*}
&& W \mbox{ does not satisfy } \psi\\
&\mbox{iff}&    \mbox{ either  } \sigma  \mbox{ or  } \tau
\mbox{ is not satisfied by } W\\
&\mbox{iff}&    \mbox{ either  }\,\,\,
\widehat{\sigma^{\ddagger}}(w\toe)=\laste\,\,\,   \mbox{ or  } \,\,\,
\widehat{\tau^{\ddagger}}(w\toe)=\laste, \\
&\mbox{whence}&\widehat{\psi^{\ddagger}}(w\toe) =
\min(\widehat{\sigma^{\ddagger}}(w\toe),
\widehat{\tau^{\ddagger}}(w\toe)) =
\laste.  
\end{eqnarray*}

%$W$ satisfies $\psi^\dagger$ iff $W$ satisfies both
%$\sigma$ and $\tau$. By induction hypothesis,
%$\widehat{\sigma^{\ddagger}}(w\toe)
%= \widehat{\tau^{\ddagger}}(w\toe) = 1$,
%whence by construction,
%$$\widehat{\psi^{\ddagger}}(w\toe)
%=(\widehat{\sigma^{\ddagger}}
%\wedge \widehat{\tau^{\ddagger}})(w\toe) =
%\min( \widehat{\sigma^{\ddagger}}(w\toe),
%\widehat{\tau^{\ddagger}}(w\toe))
%=1.$$
%Conversely, $W$ does not satisfy $\psi^\dagger$ iff
%either $\sigma$ or $\tau$ is not satisfied by
%$W$;   by induction hypothesis we either
%have $\widehat{\sigma^{\ddagger}}(w\toe)=\laste$ or
%$\widehat{\tau^{\ddagger}}(w\toe)=\laste$, whence,
%$\widehat{\psi^{\ddagger}}(w\toe) =
%\min(\widehat{\sigma^{\ddagger}}(w\toe),
%\widehat{\tau^{\ddagger}}(w\toe)) =
%\laste.$

The case  $\psi^\dagger = \sigma \vee \tau$ is similar.
\end{proof}

%%%%%%%%%%%%%%%%%%%%%%%%%%%%%%%%%%%
\section{Main results}
\label{section:final}
%%%%%%%%%%%%%%%%%%%%%%%%%%%%%%%%%%%%%%
The incorporation
  into  \L$_\infty$-formulas
   of the numerical parameters
   of the Stable Consequence problem
  relies on the following:

\begin{proposition}
\label{proposition:odot}
For  $\Phi=\{\phi_1,\dots,\phi_u\}$   a finite set of
boolean formulas
in the variables $X_1\dots,X_n,$  let
the integers  $d$ and $e$ satisfy the conditions
  $0\leq d < u$  and $e \geq\max(2,d)$.
Then
the following conditions are equivalent:

\medskip
\begin{itemize}
\item[(i)] Every subset $\Phi'$ of $\Phi$ obtained by deleting
$d$ elements of $\Phi$ is unsatisfiable.

\medskip
\item[(i')] Every subset $\Phi'$ of $\Phi$ obtained by deleting
up to $d$ elements of $\Phi$ is unsatisfiable.

\medskip
\item[(ii)]  For each   valuation $V\in \Val_n$ such that
$V(X_i)\in  \left\{\frac{1}{e+1},\frac{e}{e+1}\right\}$ for all
$i=1,\ldots, n,\,\,\,$ we have
$V
\left(\left(\bigodot_{j=1}^{u}\phi_{j}^{\ddagger}\right)
\to (X_1\vee\neg X_1)^{d+1}\right)=1.$
%(Notation of (\ref{equation:iterated})).
\end{itemize}
\end{proposition}

\begin{proof} (i)$\Leftrightarrow$(i') is trivial.
  (i') $\Rightarrow$ (ii)  Let $V$ be a counterexample to (ii).
Since for all $i=1,\ldots,n$,  $V(X_i)\in  \left\{\frac{1}{e+1},\frac{e}{e+1}\right\}$,
upon identifying  the restriction
$V\restrict \{X_1,\ldots,X_n\}$ with
the point
$$(V(X_1),\ldots,V(X_n))\in \cube$$
we can write
\begin{equation}
\label{equation:due-valutazioni}
V\restrict \{X_1,\ldots,X_n\}=(W\restrict\{X_1,\ldots,X_n\})\toe
\end{equation}
for
a unique boolean valuation
  $W$ of the set boolean formulas
  $\psi(X_1,\dots,X_n).$
Since (ii)  fails for $V$, by
definition of the implication connective in \L$_\infty$ we can write
$$
V\left(\bigodot_{j=1}^{u}\phi_{j}^{\ddagger}\right)>
V((X_1\vee\neg X_1)^{d+1}).
$$
{}From
$$
%\label{equation:estense1}
V(X_1\vee\neg X_1) = \max\left(\frac{1}{e+1},\frac{e}{e+1}\right)=
\frac{e}{e+1}
$$
we obtain by (\ref{equation:iterated}) 
and (\ref{equation:valuation1})
$$
%\label{equation:estense2}
V((X_1\vee\neg X_1)^{d+1})=1-\frac{d+1}{e+1},
$$
whence
\begin{equation}
\label{equation:estense3}
V\left(\bigodot_{j=1}^{u}\phi_{j}^{\ddagger}\right)>
1-\frac{d+1}{e+1}.
\end{equation}
Our assumption about $V$ is to the effect that
$V\left(\bigodot_{j=1}^{u}\phi_{j}^{\ddagger}\right)$
is an integer multiple of  $\frac{1}{e+1}$, whence by
(\ref{equation:estense3})  
\begin{equation}
\label{equation:estense4}
V\left(\bigodot_{j=1}^{u}\phi_{j}^{\ddagger}\right)
\geq 1-\frac{d}{e+1}\,,
\end{equation}
and by Definition \ref{definition:ddagger},
$$
V\left(\phi_{j}^{\ddagger}\right) \in \left\{\frac{e}{e+1},1\right\}, \mbox{ for all }
j=1,\ldots,u.
$$
Thus  by (\ref{equation:estense4}),   at
most $d$  among the formulas  $\phi_{1}^{\ddagger},
\dots,\phi_{u}^{\ddagger}$
are evaluated  to $\frac{e}{e+1}$ by $V$.
By (\ref{equation:due-valutazioni})
together with Propositions \ref{proposition:tutto}
and   \ref{proposition:2223},  at
most $d$  among the formulas  $\phi_{1},\ldots,
\phi_{u}$ are evaluated to 0 by $W$.
Thus, at least  $u-d$  are satisfied by $W$,
against assumption (i').

\medskip
(ii) $\Rightarrow$ (i)   If (i) fails then without loss of
generality we can assume the set  $\Phi'=\{\phi_{1},\dots,\phi_{u-d}\}$ to be
satisfiable by some boolean valuation $Y$. Let the point
$z = (Y(X_1),\ldots,Y(X_n))  \in \cantor$ be
(identified with)  the restriction of $Y$
to the set of variables  $\{X_1,\ldots,X_n\}$.
Let  $U\in \Val_n$ be the valuation  uniquely determined
by the stipulation
$$
U\restrict \{X_1,\ldots,X_n\}= z\toe.
$$
Then $U$  satisfies the hypothesis
of  (ii),
$$
U(X_i) \in \left\{\frac{1}{e+1},\frac{e}{e+1}\right\}  \mbox{ for all }
i=1,\ldots,n,
$$
whence by   (\ref{equation:iterated}) 
and (\ref{equation:valuation1}), 
$$U((X_1\vee\neg X_1)^{d+1})=1-\frac{d+1}{e+1}.$$
Since $Y$ satisfies $\Phi',$  from
Proposition \ref{proposition:2223}   we get
$$U\left(\bigodot_{j=1}^{u}\phi_{j}^{\ddagger}\right)\geq 1-\frac{d}{e+1}.$$
Thus, 
$$
U\left(\bigodot_{j=1}^{u}\phi_{j}^{\ddagger}\right)
%\geq 1-\frac{d}{e+1} 
>  1-\frac{d+1}{e+1} = U((X_1\vee\neg X_1)^{d+1}),
$$
and, by definition of the $\to$ connective, (ii) fails.
\end{proof}

% Let $\{\Phi_1,\dots,\Phi_k\}= \{\{\phi_{11},\dots,\phi_{1u(1)}\},
%\dots,\{\phi_{k1},\dots,\phi_{ku(k)}\}\}$ be a finite set of finite sets

\begin{theorem}
\label{theorem:final}  
Let $n$ and $k$ be integers $>0$.
For each
$i=1,\dots,k$ let
$\Phi_i=\{\phi_{i1},\phi_{i2},\dots,\phi_{iu(i)}\}$
be a finite set of  boolean formulas in the variables
$X_1,\dots,X_n$. Also let the integer $e_i$   satisfy
$0\leq e_i<u(i)$. Then the following conditions are equivalent:

  \medskip
\begin{itemize}
\item[(i)]   For any subset $\Phi'_i\subseteq \Phi_i$
having  $u(i)-e_i$ elements  $(i=1,\ldots,k)$, the
boolean formula  $\bigwedge_{i=1}^k \Phi'_i$ is unsatisfiable.

\medskip
\item[(ii)] In infinite-valued \luk\ logic \L$_\infty$,
letting  $e = \max(2,e_1,\dots,e_k)$ and recalling the  notation of
(\ref{equation:iterated})-(\ref{equation:iterated-disj}),
we have
$$
\bigwedge_{t=1}^n\left((X^e_t \leftrightarrow \neg X_t) \vee
(X_t \leftrightarrow \neg \,\,e\centerdot X_t)\right)
\,\vdash
\,\bigwedge_{i=1}^k\left(\left(\bigodot_{j=1}^{u(i)}\phi_{ij}^{\ddagger}\right)
\to (X_1\vee\neg X_1)^{e_i+1}\right).
$$
\end{itemize}
\end{theorem}

\medskip
\begin{proof}
Immediate from  Propositions
 \ref{proposition:woj} and
   \ref{proposition:odot}, using the
characterization in
   Proposition
\ref{proposition:due-sei}
of all  valuations satisfying
$
\bigwedge_{t=1}^n\left((X^e_t \leftrightarrow \neg X_t) \vee
(X_t \leftrightarrow \neg \,\,e\centerdot X_t)\right).
$
\end{proof}

A problem $\mathcal Q$
is said to be in coNP if its complementary
problem is in NP.  If,  in addition,
every problem in coNP is reducible
to $\mathcal Q$  in polynomial time, then
$\mathcal Q$  is {\it coNP-complete.}

\begin{corollary}
\label{corollary:final}
Fix  integers  $n, k>0$.
\begin{itemize}
\item[(i)]  For any instance
$$
I=
\left(\{\phi_{11},\dots,\phi_{1u(1)}\},
\dots,\{\phi_{k1},\dots,\phi_{ku(k)}\};  \,e_1,\ldots,e_k \right)
$$
of the Stable Consequence problem
in the variables $X_1,\ldots,X_n$, let
   $\rho(I)$ be the pair of \L$_\infty$-formulas
  $$
  \left(\bigwedge_{t=1}^n\left((X^e_t \leftrightarrow \neg X_t) \vee
(X_t \leftrightarrow \neg \,\,e\centerdot X_t)\right)
,\,\, \bigwedge_{i=1}^k\left(\left(\bigodot_{j=1}^{u(i)}\phi_{ij}^{\ddagger}\right)
\to (X_1\vee\neg X_1)^{e_i+1}\right)\right),
  $$
  where $e=\max(2,e_1,\ldots,e_k).$
Then $\rho$ reduces in polynomial time the
  Stable Consequence problem to the
  \L$_\infty$-consequence problem.

  \medskip
\item[(ii)] 
There is a  constant
$c$ such that
%letting  $v_I$  be the number of distinct variables in $I$, 
%the length
%   $|\rho(I)|$
%  of  $\rho(I)$   satisfies the inequalities
  $$
  |\rho(I)| \leq c \cdot n \cdot |I|
  < c\cdot |I|^2 \quad  \mbox{\it for all $n$  and }  I.
  $$

\medskip
\item[(iii)]  The Stable Consequence problem is
coNP-complete.
\end{itemize}
\end{corollary}

\begin{proof}
(i)
By Theorem \ref{theorem:final},
  $\rho(I)$ belongs to the \L$_\infty$-consequence problem iff
  $I$ belongs to the  Stable Consequence problem.
  Trivially,  $\rho$  is computable in polynomial time.

  (ii) These inequalities follow by direct inspection
  of the two formulas in (i).
  With reference to the notational conventions
  (\ref{equation:iterated})-(\ref{equation:iterated-disj}), 
  it should be noted that we do not have in
   \L$_\infty$ an exponentiation connective
   for $\psi^e$, nor a multiplication
   connective for $e\centerdot \psi$
   making $|\rho(I)|$ proportional to $|I|$.

\medskip
(iii)  In order to show that
  an instance
  $$I=(\Phi_1,\dots,\Phi_k;\,\,e_1,\ldots,e_k)$$
   does {\it not}
  belong to the Stable Consequence problem,
  for each $i=1,\ldots,k$
  one must guess a set  $\Delta_i\subseteq \Phi_i $ 
  with  $e_i$ elements, and a boolean valuation that
  simultaneously satisfies the conjunction of all formulas in    
  $(\Phi_1\setminus \Delta_1)\cup \cdots \cup (\Phi_k\setminus \Delta_k).$
Thus the  Stable Consequence problem
  is in coNP.

The desired coNP-completeness result now
easily follows, since the
  the Stable Consequence problem
  contains  the Unsatisfiability  problem---the prototypical
  coNP-complete problem.   Instances $I$ of
  the Unsatisfiability problem are those
  with $k=1$  and $e_1=0$.
   \end{proof}

\begin{corollary}
\label{corollary:remake}
  The  \L$_\infty$-consequence problem
  is co-NP complete.
\end{corollary}

\begin{proof}
In the light of  Corollary
   \ref{corollary:final}
there remains to be proved that
  the  \L$_\infty$-consequence problem  is in coNP.  So let
$(\theta,\phi)$  be an instance of
   the  \L$_\infty$-consequence problem,
   for some  $\phi,\theta\in \Form_n$.
   Let $L=\{l_1,\ldots,l_u\}$  be a set
   containing  the linear pieces of  $\hat\phi.$
   $L$ can be easily obtained by induction on the
   number $j'$ of connectives  occurring in $\phi$.
   The same induction shows that
   the maximum $a'$ of the absolute values of the coefficients
   of $l_1,\ldots,l_u$ satisfies the inequality
     $a'\leq j'+1$
   (actually, negation connectives have no
   effect on the value of  $a$).
   Let similarly  $M=\{l_{u+1},l_{u+2},\ldots,l_v\}$  be a  set
   containing  the linear pieces of  $\hat\theta.$
   Letting  $j''$  be the number of
connectives  in $\theta$,
   the absolute value $a''$  of the coefficients of
   all  $l_i\in M$ is bounded by   $j''+1$.
   Denoting by $a$ the maximum of the absolute values
   of the coefficients of every  $l_i\in L \cup M,$
   we can write
   \begin{equation}
   \label{equation:bound}
   a\leq j'+j'' < |\theta|+|\phi|.
   \end{equation}

   For each permutation $\phi$  of the index set
   $\{1,\ldots,v\}$ we have a (possibly empty)
   compact convex polyhedron
   $$
   P_\pi=\{x\in \cube\mid l_{\pi(1)}\leq l_{\pi(1)}\leq \cdots\leq 
l_{\pi(v)} \}.
   $$
   By construction, both $\hat\theta$ and $\hat\phi$  are
   linear over $P_\pi.$
   Now letting $\pi$ range over all possible permutations
   of   $\{1,\ldots,v\}$, the family of $P_\pi$ and their faces
   will constitute what is known as a {\it polyhedral complex}
   on $\cube.$  In other words, the union of the $P_\pi$
   is  $\cube$,  and any two polyhedra intersect in a common face.

 By Propositions \ref{proposition:woj} and
  \ref{proposition:tutto},       $\theta\nvdash\phi$
  iff
  $\hat\phi$  does not constantly take value 1 over
   $\hat\theta^{-1}(1)$ iff 
there is  a permutation $\pi$ and  a vertex $x$ of $P_\pi$
   such that   $\hat\theta(x)=1$ and $\hat\phi(x)<1.$
Such  $x$  is a rational point
$$
x=(a_1/b,\dots,a_n/b),\,\,\,a_i,b\in \mathbb Z,
\,\,\, 0\leq a_i\leq b\not=0
$$
  obtained by intersecting $n+1$
  linear functions
$l_i\in L \cup M$.  In other words, for the calculation of $x$
one must solve
  a system of $n$ linear
equations in the $n$ unknowns  $x_1,\ldots,x_n,$
where the coefficients of each equation
are integers   $\leq a$ as in   (\ref{equation:bound}).
Then
a  routine  computation using
Hadamard inequality  shows that the denominator
  $b$  of $x$ satisfies the inequality
$$
b < 2^{p(|(\theta, \phi)|)}
$$
for some fixed polynomial $p$, independent of the pair
$(\theta,\phi)$.
Writing now each coordinate $a_i/b$ of
  $x$ as a pair of
integers in decimal, or in binary notation,
we conclude that the length of $x$
  is bounded
by $q(|(\theta,\phi)|)$, for some polynomial $q$,
also independent of $(\theta,\phi)$.

Summing up, the following is a
  non-deterministic polynomial time
  decision procedure for   $\theta
  \nvdash\phi$:
\begin{itemize}
\item[---] Guess such short rational $x\in \cube$
and, proceeding bottom-up
throughout the parsing trees
of $\theta$ and $\phi$,
\item[---]  Quickly verify that $\hat\theta(x)=1$ and $\hat\phi(x)<1.$
\end{itemize}
We have thus proved  that
the  \L$_\infty$-consequence problem
is
in coNP, as required to complete the proof.
\end{proof}

%%%%%%%%%%%%%%%%%%%%%%%%%%%%%%
\section{Concluding remarks}
\label{section:remarks}
%%%%%%%%%%%%%%%%%%%%%%%%%%%%%%
Suppose the evidence at our disposal
to draw a certain conclusion $\omega$
in boolean logic
rests upon a  set
$\Theta=\{\phi_1,\ldots,\phi_{m}\}$
of  boolean formulas.
Suppose some formulas in $\Theta$  are
dubious, but  removal  of
the set $\nabla\subseteq \Theta$ of all dubious
formulas  
would dash  all hopes
to derive $\omega$  from $\Theta\setminus \nabla$.
Then $\Theta$ must be looked at
with the keener eyesight provided by
infinite-valued \luk\ logic.

Letting $\Delta=\Theta\setminus \nabla$,
any integer
$e=0,\ldots,{\rm card}(\nabla)-1,$
determines an instance
$$J_e =(\Delta\cup\{\neg\omega\},\nabla;\,\, 0,e)$$
  of the Stable Consequence problem,
together with its associated pair
$\rho(J_e)=(\theta_e,\phi_e)$
of \L$_\infty$-formulas.
While   $e$
measures no individual property of formulas in  $\nabla,$
    it makes perfect
   sense to ask whether
$\omega$  invariably 
follows  from  $\Delta\cup \nabla',$
for each set
$\nabla'\subseteq \nabla$ obtained by randomly expunging
up to $e$   formulas of $\nabla$.
By Theorem \ref{theorem:final},  the condition
$\theta_e\vdash \phi_e$ 
holds iff a fraction 
$$0\leq \frac{e}{\mbox{card}(\nabla)}<1$$
of dubious formulas can be randomly
removed  without prejudice
to our deduction of $\omega$ from $\Delta$
and the rest of $\nabla$ in boolean logic.

Generalizing the Maximum Satisfiability problem,
let $e^*$  be the largest integer $e$
such that $\theta_e\vdash \phi_e$.
If we strongly doubt about $\nabla$
then $\frac{e^*}{\mbox{card}(\nabla)}$ should
be close to 1, meaning that $\omega$  can be safely
obtained even if we randomly dismiss most formulas of
$\nabla$.  On the other hand,
  when the formulas in $\nabla$
are almost as sound as those in $\Delta$,
  we can afford a small value of
$\frac{e^*}{\mbox{card}(\nabla)}$,
telling us that
almost all formulas in $\nabla$ are necessary
to draw $\omega.$

Binary search yields such $e^*$ after  checking
   $\theta_e\vdash \phi_e$ for only
O$(\log_2({\rm card}(\nabla)))$
different values of $e$.
Any such instance of
the \L$_\infty$-consequence problem
translates into purely logical terms
  the imprecisely stated
problem whether the deduction of $\omega$
in boolean logic
essentially, inessentially, substantially,
marginally, critically,  \ldots relies on $\nabla.$

\end{document}